\documentclass[letterpaper, 10 pt, conference]{ieeeconf}  

\IEEEoverridecommandlockouts                              
\usepackage{mathtools}
\usepackage{amsmath}
\usepackage{amssymb}
\usepackage{mathrsfs}
\usepackage{graphicx}
\usepackage{subfigure}
\usepackage{color}
\usepackage{bm}
\usepackage{array}
\usepackage{algorithm}
\usepackage{algorithmic}

\usepackage{savesym}
\savesymbol{labelindent}
\usepackage{enumitem}
\restoresymbol{enumitem}{labelindent}

\makeatletter
\let\NAT@parse\undefined
\makeatother
\usepackage[hidelinks]{hyperref}

\allowdisplaybreaks[4]

\newtheorem{asmp}{Assumption}
\newtheorem{lem}{Lemma}
\newtheorem{prop}{Proposition}
\newtheorem{thm}{Theorem}

\newtheorem{rmk}{Remark}
\newtheorem{defn}{Definition}

\allowdisplaybreaks[4]

\newcommand{\RR}{\mathbb{R}}
\newcommand{\EE}{\mathbb{E}}
\newcommand{\PP}{\mathbb{P}}
\newcommand{\NN}{\mathbb{N}}

\newcommand{\II}{\mathbb{I}}
\newcommand{\eps}{\varepsilon}

\newcommand{\mtcc}{\mathcal{C}}

\newcommand{\mtce}{\mathcal{E}}

\newcommand{\mtcg}{\mathcal{G}}

\newcommand{\mtcv}{\mathcal{V}}

\newcommand{\bfb}{\mathbf{b}}

\newcommand{\bff}{\mathbf{f}}

\newcommand{\bfn}{\mathbf{n}}

\newcommand{\bfv}{\mathbf{v}}
\newcommand{\bfw}{\mathbf{w}}
\newcommand{\bfx}{\mathbf{x}}
\newcommand{\bfy}{\mathbf{y}}

\newcommand{\bfF}{\mathbf{F}}

\newcommand{\bfS}{\mathbf{S}}

\newcommand{\bfX}{\mathbf{X}}
\newcommand{\bfY}{\mathbf{Y}}

\newcommand{\bfPhi}{\mathbf{\Phi}}

\newcommand{\Let}{: =}
\newcommand{\teL}{= :}
\newcommand{\sgn}{\textup{sgn}}

\newcommand{\bfl}{\mathbf{1}}
\newcommand{\bfo}{\mathbf{0}}
\newcommand{\tx}{\textup}
\newcommand{\tp}{\tx{T}}

\newcommand{\tsbm}{\tx{SBM}}
\newcommand{\tssbm}{\tx{SSBM}}
\newcommand{\rank}{\operatorname{rank}}
\newcommand{\diag}{\operatorname{diag}}
\newcommand{\range}{\operatorname{range}}

\title{\LARGE \bf
Learning Communities from Equilibria of Nonlinear Opinion Dynamics
}

\author{Yu Xing, Anastasia Bizyaeva, and Karl H. Johansson
\thanks{This work was supported by the Knut and Alice Wallenberg Foundation (Wallenberg Scholar Grant), the Swedish Research Council (Distinguished Professor Grant 2017-01078), and the Swedish Foundation for Strategic Research (SUCCESS FUS21-0026).}
\thanks{YX and KHJ are with Division of Decision and Control Systems, School of Electrical Engineering and Computer Science, KTH Royal Institute of Technology, and with Digital Futures, Stockholm, Sweden. Email:
		{\tt\small \{yuxing2, kallej\}@kth.se}.
  AB is with the Sibley School of Mechanical and Aerospace Engineering, Cornell University, Ithaca, NY, USA. Email: {\tt\small anastasiab@cornell.edu}.}%
}

\begin{document}

\maketitle
\thispagestyle{empty}
\pagestyle{empty}

\begin{abstract}
This paper studies community detection for a nonlinear opinion dynamics model from its equilibria. It is assumed that the underlying network is generated from a stochastic block model with two communities, where agents are assigned with community labels and edges are added independently based on these labels. Agents update their opinions following a nonlinear rule that incorporates saturation effects on interactions. It is shown that clustering based on a single equilibrium can detect most community labels (i.e., achieving almost exact recovery), if the two communities differ in size and link probabilities. When the two communities are identical in size and link probabilities, and the inter-community connections are denser than intra-community ones, the algorithm can achieve almost exact recovery under negative influence weights but fails under positive influence weights. Utilizing fixed point equations and spectral methods, we also propose a detection algorithm based on multiple equilibria, which can detect communities with positive influence weights. Numerical experiments demonstrate the performance of the proposed algorithms.
\end{abstract}

\section{Introduction}

Learning networks from group dynamics has gained significant interest in various disciplines~\cite{ramakrishna2020user,ravazzi2021learning}, due to its wide applications in influence maximization~\cite{kempe2003maximizing} and recommender systems~\cite{burke2002hybrid}. Community detection is one of the central topics among network inference~\cite{fortunato2016community}, as real networks often comprise communities that are sparsely connected to each other. Recently, an increasing amount of research has focused on community detection based on observations from dynamical systems (e.g.,~\cite{wai2019blind,hoffmann2020community,xing2023community}). However, most studies address the detection problem for linear dynamics, and there is still a need to understand how nonlinearity affects detection performance. 
Since nonlinear models can behave differently from linear averaging dynamics~\cite{baumann2020modeling,bizyaeva2022nonlinear}, it is important to investigate how to adapt and apply traditional methods, such as spectral clustering, to complex nonlinear dynamics.

\subsection{Related Work}
Community detection has been extensively investigated for over two decades~\cite{fortunato2016community}.
There are three major approaches to the problem. The most common approach involves optimizing quality functions. The Louvain method~\cite{blondel2008fast} is a widely-used fast algorithm based on greedy optimization of the modularity, which measures the extent a given network partition implying dense connections within subgroups.
The Infomap method~\cite{rosvall2008maps} represents the approach based on dynamics. The method aims to compress random walks over networks, by looking for partitions that reduce information required for describing the trajectories.
The approach based on statistical inference has become well-established in recent years. The methods infer generative network models that presuppose community structure. The stochastic block model (SBM) is a crucial example, where edges exist with probability depending on pre-assigned community labels. This framework facilitates theoretical analysis of community detectability~\cite{yun2014community,massoulie2014community,abbe2017community}. 

When only data from dynamics over a network are available, rather than direct edge information, a standard method is to construct a network based on state similarity and then apply detection algorithms to that network. 
The paper~\cite{hoffmann2020community} introduces a Bayesian hierarchical model for time series, and demonstrates that model-based approaches can outperform traditional methods.
Maximum likelihood methods for cascade dynamics are explored in~\cite{prokhorenkova2022less}, and nonparametric Bayesian methods for epidemics and an Ising model are proposed in~\cite{peixoto2019network}. 
A blind community detection approach is developed in~\cite{wai2019blind,schaub2020blind,wai2022community}. The method applies spectral clustering to sample covariance matrices derived from a single snapshot from multiple trajectories.
The papers~\cite{xing2023community,xing2023almost} propose algorithms leveraging transient and asymptotic behaviors of gossip opinion dynamics.

Most of the above studies focus on community detection for linear dynamics, and the problem for nonlinear dynamics remains unexplored.
In this paper, we study a nonlinear opinion dynamics model~\cite{bizyaeva2022nonlinear}. 
In the model, individuals update their opinions according to a saturation interaction rule, which is also found in neural and biological systems~\cite{eldar2013effects,bogacz2007extending}.
This type of models can capture the transition from opinion consensus to polarization without the presence of external influence, unlike the DeGroot, Friedkin-Johnsen, or Hegselmann-Krause models~\cite{proskurnikov2017tutorial}.
In particular,~\cite{leonard2021nonlinear} uses the nonlinear model to explain political polarization dynamics.
Additionally, ~\cite{baumann2020modeling} proposes a model combining the saturation rule with homophily, and reproduces the echo chamber phenomenon on social media.
Therefore, investigating community detection from these nonlinear dynamics can help understand the structure of real-world dynamics.

\subsection{Contributions}
We study community detection based on equilibria of a nonlinear opinion dynamics. It is assumed that the dynamics evolve over networks generated from an SBM with two communities. For the case with a single equilibrium available, we propose a community detection algorithm based on $k$-means (Algorithm~\ref{alg_1}). The algorithm can detect most community labels with high probability (i.e., achieving almost exact recovery), if a non-consensus equilibrium is used, and the two communities differ in size and link probabilities (Theorem~\ref{thm_main}~(i.a)). 
When the communities are identical in size and link probabilities, and inter-community connections are denser than intra-community ones, the algorithm can achieve almost exact recovery if interpersonal influence weights are negative, but fails if the weights are positive (Theorem~\ref{thm_main}~(i.b) and~(ii)).
For multiple equilibria with external inputs, another detection algorithm (Algorithm~\ref{alg_2}) is developed by leveraging fixed point equations and spectral clustering methods. Its performance is validated by numerical experiments.

By studying a typical nonlinear model~\cite{bizyaeva2021control}, the results demonstrate how community detectability is affected by nonlinearity.
Stationary states under external excitation or transient states are necessary for detection in linear dynamics (e.g., \cite{wai2019blind,schaub2020blind,xing2023almost}). 
In contrast, it is found that community detection is possible by using only equilibria of nonlinear dynamics without excitation. 
However, community structure information may not be preserved in the equilibria when agents are strongly influence by nonlinear interactions. 
These findings provide key insight into the design of community detection methods for real-world complex dynamics.

\subsection{Outline}
Section~\ref{sec_prel} introduces the nonlinear dynamics and the SBM.
Section~\ref{sec_problem} formulates the problem. 
Two detection algorithms are proposed in Section~\ref{sec_alg}, and numerical experiments presented in Section~\ref{sec_simul}.
Section~\ref{sec_conclusion} concludes the paper.~\\

\noindent \textbf{\emph{Notation}.}
%
Denote the set of positive integers by $\NN_+$ and the set of positive real numbers by $\RR_+$. 
%
%
Denote the $n$-dimensional all-one vector by $\bfl_n$. 
$I_n$ is the identity matrix, and $\bfl_{m,n}$ ($\bfo_{m,n}$) is the $m\times n$ all-one (all-zero) matrix.
%
%
Denote the Euclidean norm of a vector and the spectral norm of a matrix by $\|\cdot\|$.
A vector is denoted by a boldface letter, e.g., $\bfx$, and its $i$-th entry by $x_i$. 
For a matrix $A \in \mathbb{R}^{n\times n}$, $a_{ij}$ or $[A]_{ij}$ denotes its $(i,j)$-th entry.
For a symmetric $A \in \RR^{n\times n}$, denote its smallest and largest eigenvalues by $\lambda_{\min}(A)$ and $\lambda_{\max}(A)$. 
For a vector-valued function $\bff(\bfx):\RR^n \to \RR^m$, its Jacobian matrix is an $m\times n$ matrix, denoted by $D_{\bfx} \bff$, whose $(i,j)$-th entry is $[D_{\bfx} \bff]_{ij} = \partial f_i/\partial x_j$. The Jacobian of $\bff$ at a point $\bfx_0 \in \RR^n$ is written as $D_{\bfx} \bff (\bfx_0)$.
For real numbers  $a(n),b(n) > 0$, $n\in\NN$, denote $a(n) = \Theta(b(n))$, if there exist $C_1, C_2>0$ such that $C_1 b(n) \le a(n) \le C_2 b(n)$.
%
The function $\II_{[\textup{property}]}$ is the indicator function, which is one if the property in the bracket holds, and zero otherwise. 
%
%
A sequence of events $\{A_n\}$ happens with high probability (w.h.p.) if $\lim_{n\to\infty} \PP\{A_n\} = 1$.
An undirected graph is denoted by $\mtcg = (\mtcv,\mtce,A)$, where $\mtcv$ is the agent set, $\mtce$ is the edge set, and $A = [a_{ij}]$ is the adjacency matrix such that $a_{ij} = 1$ ($a_{ij} = 0$) if $\{i,j\}\in \mtce$ ($\{i,j\}\not\in\mtce$).

\section{Preliminaries}\label{sec_prel}

In this section, we introduce the nonlinear opinion dynamics and the SBM, and briefly discuss their properties.

\subsection{Nonlinear Opinion Dynamics}
The nonlinear opinion dynamics model takes place over an undirected graph $\mtcg = (\mtcv,\mtce, A)$ with $\mtcv = \{1,\dots,n\}$ and no self-loops ($a_{ii}=0$). 
Each agent $i \in \mtcv$ has a state $x_i(t)$, $t\in \RR_+$, and the model evolves in continuous time according to the  update rule
\begin{align}\label{eq_opinion_model_agentform}
    \dot{x}_i = - d x_i + u S\Big(\alpha x_i + \gamma \sum_{k\in\mtcv} a_{ik} x_k\Big) + b_i, 
\end{align}
where $d > 0$ is the damping coefficient, $u$ is the agent attention parameter to the nonlinear network interaction, and $S$ is an odd saturating function satisfying $S(0)=0$, $S^{\prime}(0)=1$, and $\sgn(S^{\prime\prime}(z)) = - \sgn(z)$. Here we assume $S= \tanh$, i.e., the hyperbolic tangent. 
In the nonlinear term, $\alpha \ge 0$ is the self weight and $\gamma \in \RR$ is the influence weight of other agents. 
Lastly, $b_i \in \RR$ is an additive input, which can be seen as individual prejudice or external influence. 

The compact form of~\eqref{eq_opinion_model_agentform} can be written as
\begin{align}\label{eq_opinion_model_compactform}
    \dot{\bfx} = - d \bfx + u \bfS (\alpha \bfx + \gamma A \bfx ) + \bfb, 
\end{align}
where $\bfS(\bfx) \Let [S(x_1),\dots,S(x_n)]^\tp$ for $\bfx \in \RR^n$.

The model and its extensions have been thoroughly studied, focusing on their bifurcation and steady-state behavior~\cite{bizyaeva2022nonlinear,bizyaeva2021control,bizyaeva2021patterns}. 
Here we consider the case where the parameters are homogeneous (i.e., $d$, $u$, $\alpha$, and $\gamma$ are identical for all agents) and each agent has a single opinion. 
The following result \cite[Theorem~1]{bizyaeva2021patterns} demonstrates that two new equilibria of the model~\eqref{eq_opinion_model_compactform} without inputs emerge, due to bifurcation from the origin as the attention parameter $u$ increases beyond specific thresholds.

\begin{prop}\label{prop_bifurcation}
    Suppose that $\mtcg$ is connected, $u\ge 0$, and $\bfb = \bfo$. 
     \begin{enumerate}[leftmargin=*, itemsep=0pt, parsep=0pt,label=(\roman*)]
        \item If $\gamma > 0$, the origin $\bfx = \bfo$ is a locally exponentially stable equilibrium for $0<u<u_1$ and unstable for $u > u_1$, where $u_1 \Let d/(\alpha + \gamma \lambda_{\max}(A))$. At $u=u_1$, branches of equilibria $\tilde{\bfx} \not= \bfo$ emerge in a steady-state bifurcation off of $\bfx = \bfo$ along the eigenspace corresponding to $\lambda_{\max}(A)$, where the entries of $\tilde{x}$ have the same sign.

        \item If $\gamma<0$, the origin $\bfx = \bfo$ is a locally exponentially stable equilibrium for $0<u<u_2$ and unstable for $u > u_2$, where $u_2 \Let d/(\alpha + \gamma \lambda_{\min}(A))$. At $u=u_2$, branches of equilibria $\tilde{\bfx} \not= \bfo$ emerge in a steady-state bifurcation off of $\bfx = \bfo$ along the eigenspace corresponding to $\lambda_{\min}(A)$, where the entries of $\tilde{\bfx}$ have different signs. \hfill\QED
    \end{enumerate}
 \end{prop}

The real-valued opinion $x_i$ represents the agent $i$'s level of support for two options. 
The sign $\sgn(x_i)$ indicates which of the two options the agent supports, and $x_i = 0$ represents a neutral position. 
Proposition~\ref{prop_bifurcation} shows that agreement steady states with all agents having the same sign emerge from the neutral state as $u$ increases, if the influence weight is positive. 
In contrast, if the influence is negative, disagreement steady states emerge. 
%


\subsection{Stochastic Block Model}
Assume that the agent set $\mtcv$ consists of two disjoint communities $\mtcv_{1}$ and $\mtcv_{2}$. 
Let the community structure vector be $\mtcc \in \{1,2\}^n$, satisfying that $\mtcc_i = 1$ if $i \in \mtcv_{1}$ and $\mtcc_i = 2$ if $i \in \mtcv_{2}$ (i.e., agents in $\mtcv_{1}$ (in $\mtcv_{2}$) have the label~$1$ (label~$2$)). 
The two-community SBM is defined as follows.
\begin{defn}[SBM]\label{def:sbm}
	Let $n\in \NN_+$ be the network size, $\bfn = [n_1~n_2]^\tp\in \NN_+^2$ be the community size vector with $n_1+n_2=n$, and 
    \[\bm{\ell} = \begin{bmatrix}
    \ell_{11} & \ell_{12}\\
    \ell_{21} & \ell_{22}
    \end{bmatrix} \in [0,1]^{2\times 2}
    \]
    be the link probability matrix with $\ell_{12} = \ell_{21}$. 
    In $\tsbm(\bfn,\bm{\ell})$, agents $1,\dots, n_1$ are assigned with community label~$1$ and agents $n_1+1,\dots, n$ with label~$2$. %
    Then the SBM generates an undirected graph $\mtcg = (\mtcv,\mtce,A)$ without self-loops, by independently adding $\{i,j\}$ with $i\not=j$ to $\mtce$ with probability $\ell_{\mtcc_i,\mtcc_j}$. 
    If $n_1=n_2 = n/2$ and $\ell_{11} = \ell_{22}$, the SBM is called symmetric SBM, denoted by $\tssbm(n,\bm{\ell})$. 
    In this case, denote $\ell_{\tx{s}} \Let \ell_{11} = \ell_{22}$ and $\ell_{\tx{d}} \Let \ell_{12} = \ell_{21}$. \hfill\QED
\end{defn}


The following assumption on link probability $\ell_{ij}$ of an SBM is given to ensure the random graph is connected w.h.p. A technical assumption for the SSBM is also introduced.
\begin{asmp}\label{asmp_connect}
    For the $\tsbm(\bfn,\bm{\ell})$, assume that $\ell_{ij} = \omega(\log n/n)$, $\forall i,j\in \{1,2\}$. If the SBM is $\tssbm(n,\bm{\ell})$ and $\ell_{\tx{s}} > \ell_{\tx{d}}$, further assume that $\ell_{\tx{d}} = \omega(\sqrt{\ell_{\tx{s}}\log n})$.
\end{asmp}

To measure the performance of an algorithm detecting communities of the SBM from observations, we introduce the accuracy of an estimate $\hat{\mtcc}$ given by the algorithm,
\begin{align}\label{eq_accdefn}
    \tx{Acc}(\mtcc,\hat{\mtcc}) \Let \frac1n \max\bigg\{\sum_{i=1}^n \II_{[\mtcc_i = \hat{\mtcc}_i]}, \sum_{i=1}^n \II_{[\mtcc_i = 3-\hat{\mtcc}_i]} \bigg\}.
\end{align}
%
%
%

Now we define almost exact recovery of an algorithm detecting communities in SBMs as follows~\cite{yun2014community,abbe2017community}.
\begin{defn}\label{defn_almost_recovery}
    For an SBM with $n$ agents and a community structure $\mtcc$, suppose that a detection algorithm outputs an estimate $\hat{\mtcc}$. The algorithm achieves almost exact recovery, if 
    $\PP \{ \tx{Acc}(\mtcc,\hat{\mtcc}) = 1 - o(1) \} = 1 - o(1)$.\hfill\QED
\end{defn}
Almost exact recovery means that the algorithm can correctly detect most community labels w.h.p.

\section{Problem Formulation}\label{sec_problem}
We investigate how to detect communities from equilibria of the model~\eqref{eq_opinion_model_compactform}. It is assumed that a graph $\mtcg = (\mtcv,\mtce,A)$ is generated from an SBM and then fixed. The opinion model~\eqref{eq_opinion_model_compactform} evolves over this graph and reaches steady state. Since we will consider cases where the network size is large, we set $\gamma = \pm 1/\Delta$, where $\Delta \Let \max_i \EE\{\sum_j a_{ij}\}$ is the maximum expected degree of $\mtcg$. 
Hence, $\gamma A \bfx$ is a weighted average of the opinions $\bfx$.
%
The problem studied in this paper is described as follows.

\textbf{Problem.} Given a single equilibrium $\bfx^*$ or multiple input-equilibrium pairs $\{[\bfb^{(1)}~\bfx^{(1)}], \dots, [\bfb^{(m)}~\bfx^{(m)}]\}$ of the model~\eqref{eq_opinion_model_compactform} over an SBM, design algorithms to detect the communities of the SBM and analyze their performance.


\section{Detection Algorithms and Main Results}\label{sec_alg}

In this section, we first address the community detection problem based on a single equilibrium. The performance of the proposed algorithm is analyzed theoretically for several SBMs (Theorem~\ref{thm_main}). We then design a detection algorithm based on multiple equilibria, by approximating the adjacency matrix using fixed point equations of the system and applying spectral clustering techniques to the estimated matrix.

\subsection{Detection from Single Equilibrium}\label{sec_alg_single}
Given only one equilibrium, we employ the $k$-means method to cluster the states, as shown in Algorithm~\ref{alg_1}. It is assumed that the exact equilibrium is obtained. We will study detection from noisy observations in the future.

The neutral steady state $\bfx = \bfo$ provides no information about the community structure, similar to the consensus of linear dynamics such as the DeGroot model. However, the equilibria emerging at the bifurcation from the origin as stated in Proposition~\ref{prop_bifurcation} can reveal information of the network, as shown in the following theorem.

\begin{algorithm}[t]
\caption{(Detection Based on Single Equilibrium)}
\label{alg_1}
\small \textbf{Input:} Community number $k=2$.\\
\textbf{Output:}~Community estimate $\hat{\mtcc}$. 
\begin{algorithmic}[1]
\STATE{Obtain an equilibrium $\bfx^*$ of the model~\eqref{eq_opinion_model_compactform}.}
\STATE{Apply $k$-means to $\bfx^*$ to get an estimate of the community structure $\mtcc$.}
\end{algorithmic}
\end{algorithm}

\begin{thm}\label{thm_main}
    Suppose that Assumption~\ref{asmp_connect} holds, $\bfb = \bfo$, and the equilibrium $\bfx^*$ emerges at the bifurcation from the origin as stated in Proposition~\ref{prop_bifurcation}.
    \begin{enumerate}[leftmargin=*, itemsep=0pt, parsep=0pt,label=(\roman*)]
        \item Assume that $\gamma = 1/\Delta > 0$, and $u-\bar{u}_1>0$ is small enough, where $\bar{u}_1 = d/(\alpha + \gamma \lambda_{\max}(\EE\{A\}))$.
        \begin{enumerate}[label=(i.\alph*)]
            \item For $\tsbm(\bfn,\bm{\ell})$, if $n_2 =o(n_1)$, $\ell_{11} n_1 = \Theta(\ell_{22} n_2)$, and $\ell_{12} = \Theta(\sqrt{\ell_{11}\ell_{22}})$, then Algorithm~\ref{alg_1} achieves almost exact recovery.
            \item For $\tssbm(n,\bm{\ell})$, it holds that $\|\bfx^* - \theta c(u) \bfl_n/\sqrt{n}\| = o(c(u))$ w.h.p., for some $\theta \in \{1,-1\}$. Here $c(u) \in \RR$ depends on $u$ and $c(u) \to 0$ as $u\to \bar{u}_1$. If $\liminf \ell_{\tx{d}}/\ell_{\tx{s}} > 3$, then Algorithm~\ref{alg_1} cannot achieve almost exact recovery.
        \end{enumerate}
        \item Assume that $\gamma = - 1/\Delta < 0$, and $u-\bar{u}_2>0$ is small enough, where $\bar{u}_2 = d/(\alpha+\gamma\lambda_{\min}(\EE\{A\}))$. For $\tssbm(n,\bm{\ell})$, if $\ell_{\tx{s}} < \ell_{\tx{d}}$, then Algorithm~\ref{alg_1} achieves almost exact recovery. \hfill\QED
    \end{enumerate}
\end{thm}

\begin{rmk}
    The result~(i.a) states that the communities of the SBM can be detected, if the influence weights are positive and the two communities differ in size and link probabilities.
    The result~(i.b) shows that the entries of the equilibrium $\bfx^*$ are close to each other in the SSBM. 
    As a result, the communities cannot be detected by clustering.
    However, almost exact recovery in the disassortative SSBM (inter-community connections are denser than intra-community) can be achieved, if the influence weight is negative. 

    From Proposition~\ref{prop_bifurcation}, the equilibrium $\bfx^*$ reflects the eigenvector centrality if the influence weight is positive. Note that agents in the  SSBM have similar centrality. Conditions in~(i.a) imply a leader-follower structure, where the community~$\mtcv_2$ has much fewer agents but denser connections, compared with the community~$\mtcv_1$. This structure ensures different centrality in communities, making detection possible. In the case of negative influence, the equilibrium $\bfx^*$ is close to the eigenvector corresponding to $\lambda_{\min}(A)$. For the disassortative SSBM, the sign of entries of that eigenvector corresponds to community labels, so almost exact recovery can be achieved.
    
    In the theorem, we assume that $u$ is close to the bifurcation threshold $\bar{u}_1$ or $\bar{u}_2$. To ensure detectability, $u$ has to be neither too small or too large. When $u$ is below the bifurcation threshold, only the neutrality equilibrium exists, providing no information about the graph. As $u$ increases, agents become more attentive to nonlinear interactions, and equilibria move away from the eigenvector centrality. See Section~\ref{sec_simul} for how $u$ and different sigmoid functions influence detection.
    \hfill\QED
\end{rmk}

\quad\emph{Proof Sketch:}
    We briefly explain how the results are obtained. 
    The detailed proof is given in Appendix~\ref{app_pf_thm}. 
    First, leveraging the Lyapunov-Schmidt reduction, we can calculate an explicit form of the equilibrium
    . The equilibrium is close to (up to a sign flip) the eigenvector corresponding to the largest or smallest eigenvalue of the adjacency matrix $A$. Next, from matrix perturbation theory and concentration inequalities, it follows that the aforementioned eigenvector is close to the eigenvector of the expected graph
    . Analyzing the properties of the expected graph and eigenvectors completes the proof.  \hfill\QED 

We conclude this subsection by discussing several extensions of Theorem~\ref{thm_main}. The result in~(i.a) holds as long as agents in different communities have different eigenvector centrality. 
The disassortative condition for $\ell_{\tx{d}}$ in~(i.b) is technical, and simulation given in Section~\ref{sec_simul} shows that Algorithm~\ref{alg_1} also fails to achieve almost exact recovery for the assortative SBM. 
The theorem only considers all-positive ($\gamma>0$) or all-negative ($\gamma<0$) relationships. A natural extension is that agents within the same community have positive edges, whereas agents between the two communities have negative edges. 
In this signed-graph case, equilibria of different communities have opposite signs (see e.g., \cite[Theorem~1]{bizyaeva2022switching}), so communities are detectable.

\begin{algorithm}[t]
\caption{(Detection Based on Multiple Equilibria)}
\label{alg_2}
\small \textbf{Input:} Multiple input-equilibrium pairs $\{[\bfb^{(1)}~\bfx^{(1)}],$ $ \dots, [\bfb^{(m)}~\bfx^{(m)}]\}$ with $m\in \NN_+$, model parameters $\alpha$, $\gamma$, $d$, $u$, and community number $k = 2$.\\
\textbf{Output:}~Community estimate $\hat{\mtcc}$. 
\begin{algorithmic}[1]
\STATE{Calculate for $1\le k\le m$
\begin{align*}
    \bfy^{(k)} = \frac1\gamma \Big(\bfS^{-1} \Big(\frac1u (d \bfx^{(k)} - \bfb^{(k)}) \Big) - \alpha \bfx^{(k)} \Big).
\end{align*}
}
\STATE{Calculate the estimate $\hat{A}$ of the adjacency matrix $A$:
\begin{align*}
    \tilde{A} &= \bfY \bfX^{\dagger},\\
    \hat{A} &= (\tilde{A} + \tilde{A}^\tp)/2,
\end{align*}
where
\begin{align*}
    \bfX &= [\bfx^{(1)}~\cdots~\bfx^{(m)}],\\
    \bfY &= [\bfy^{(1)}~\cdots~\bfy^{(m)}].
\end{align*}
}
\STATE{Find the eigenvector corresponding to the second largest eigenvalue of $\hat{A}$, denoted by $\hat{\bfv}$. Apply $k$-means to $\hat{\bfv}$ and obtain an estimate of community labels $\hat{\mtcc}$.
}
\end{algorithmic}
\end{algorithm}

\begin{figure*}[tbp]
    \centering    \subfigure[\label{fig_eep_acc}Accuracy of Algorithm~\ref{alg_1} increases with $n$ but decreases with $u$.]{\quad~~ \includegraphics[width=0.22\textwidth]{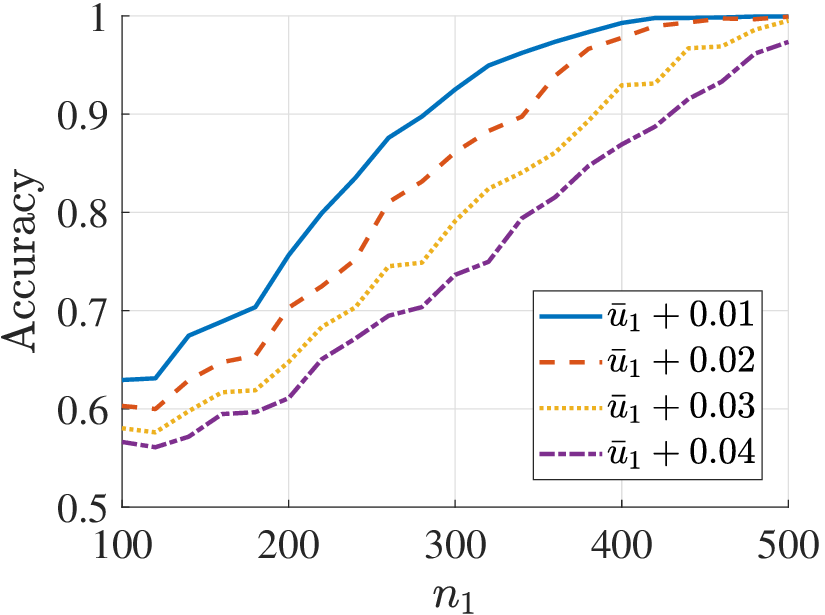} \quad~~
	}~~    \subfigure[\label{fig_eep_nonlinear}Accuracy of Algorithm~\ref{alg_1} for the model~\eqref{eq_opinion_model_compactform} with $S(x)$ being different nonlinear functions, where $u = \bar{u}_1 + 0.04$.]{\quad~~\includegraphics[width=0.22\textwidth]{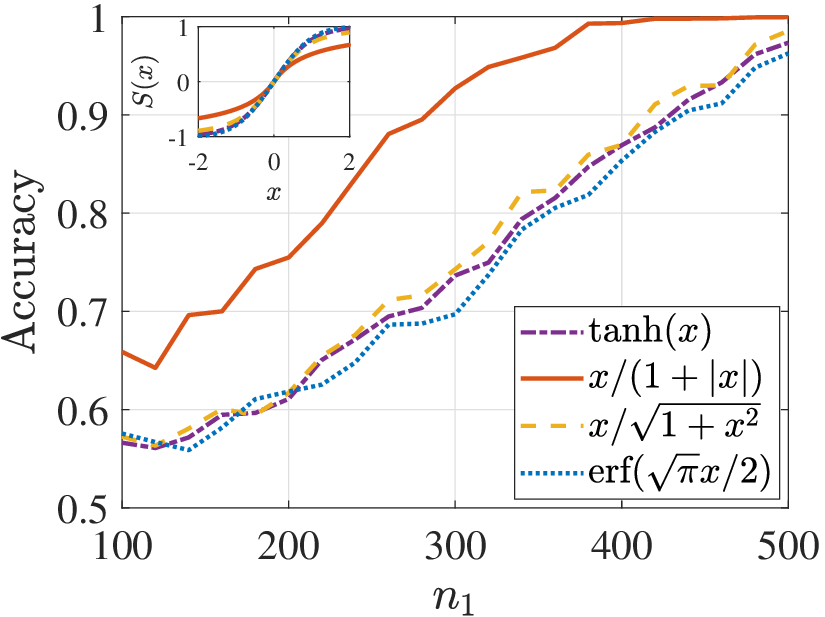} \quad~~
	} ~~   \subfigure[\label{fig_eep_comp}Accuracy of Algorithm~\ref{alg_1} with $u = \bar{u}_1+0.01$, the Louvain, the Girvan-Newman, and spectral clustering methods.]{\quad~~\includegraphics[width=0.22\textwidth]{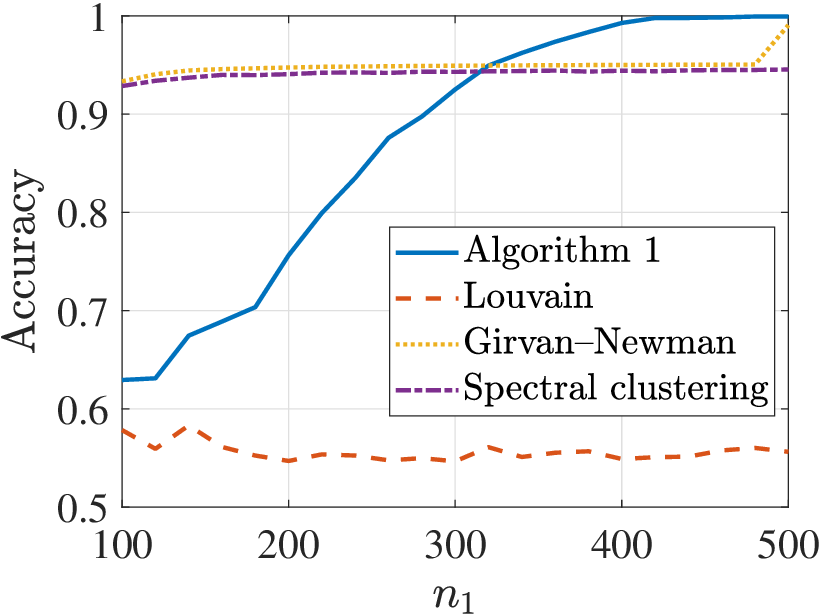} \quad~~
	}

    \caption{\label{fig_eep}Detection accuracy for an SBM with unequal-sized communities and $\gamma > 0$.}
\end{figure*}

\subsection{Detection from Multiple Equilibria}\label{sec_alg_multiple}

In this subsection, we assume that the system has external inputs, and the resulting equilibria from multiple trajectories are available. Consider the dataset of input-equilibrium pairs $\{[\bfb^{(1)}~\bfx^{(1)}], \dots, [\bfb^{(m)}~\bfx^{(m)}]\}$ with $m\in \NN_+$. We will design a community detection algorithm based on fixed point equations of the model.

The multiple trajectory case is a commonly studied scenario, in which different discussions are observed (e.g.,~\cite{ravazzi2021learning,wai2019blind,schaub2020blind}). The inputs can be seen as individual prejudice towards different topics or external information that is given for each discussion. 

Note that an input-equilibrium pair $[\bfb^{(k)}~\bfx^{(k)}]$, $1\le k\le m$, satisfies the following fixed point equation 
\begin{align*}
    \bfo = - d \bfx^{(k)} + u \bfS (\alpha \bfx^{(k)} + \gamma A \bfx^{(k)} ) + \bfb^{(k)},
\end{align*}
which implies that
\begin{align*}
    \bfy^{(k)} \Let \frac1\gamma \Big(\bfS^{-1} \Big(\frac1u (d \bfx^{(k)} - \bfb^{(k)}) \Big) - \alpha \bfx^{(k)} \Big) = A \bfx^{(k)},
\end{align*}
where $\bfS^{-1}(\bfx) \Let [S^{-1}(x_1),\dots,S^{-1}(x_n)]^\tp$ for $\bfx\in\RR^n$. Hence, 
$
    \bfY\Let [\bfy^{(1)}~\cdots~\bfy^{(m)}] = A [\bfx^{(1)}~\cdots~\bfx^{(m)}].
$
It is expected that $\bfX \Let [\bfx^{(1)}~\cdots~\bfx^{(m)}]$ is invertible when $m$ is large and $\{\bfb^{(k)}\}$ sufficiently excites the system. Then the estimate of the adjacency matrix $A$ can be given by $\bfY \bfX^{-1}$. When the number of samples is much less than the network size,  $\hat{A} = \bfY \bfX^{\dagger}$ gives an approximation of the matrix $A$, where $\bfX^{\dagger}$ is the pseudo-inverse of $\bfX$. 

Inspired by this observation, we propose Algorithm~\ref{alg_2} which utilizes spectral clustering techniques to recover the communities for the SSBM. 
Line~$1$ of the algorithm calculates the data matrix $\bfY$. In Line~$2$, $\bfY \bfX^{\dagger}$ is calculated and then projected to the set of symmetric matrices, since $A$ is symmetric. For the adjacency matrix of an SSBM, its expectation has a block structure, and the eigenvector corresponding to its second largest eigenvalue satisfies that the entries in different communities have different signs. Leveraging this property, Line~$3$ applies the spectral clustering method. Knowing the model and parameters is a strong assumption, and future work will explore how to address this limitation.

As shown in Theorem~\ref{thm_main}~(i.b), the SSBM may not be recovered by using a single equilibrium without external inputs. However, when multiple input-equilibrium pairs are available, the recovery is possible, as shown in Section~\ref{sec_simul}.

\section{Numerical Experiments}\label{sec_simul}

In this section, we conduct numerical experiments for community detection in the model~\eqref{eq_opinion_model_compactform}. For all experiments, we set the damping coefficient $d=1$, the agent self weight $\alpha = 1$, the positive influence weight $\gamma = 1/\Delta$, and the negative weight $\gamma = -1/\Delta$, where $\Delta$ is the maximum expected degree of the SBM. The equilibria are obtained using the \texttt{ode45} solver in MATLAB. The results are consistent with other solvers such as \texttt{ode78}, \texttt{ode89}, and \texttt{ode15s}.

We first study the SBM that has two communities different in size and link probabilities, as in Theorem~\ref{thm_main}~(i.a). In this experiment, we consider positive interpersonal influence ($\gamma > 0$), and calculate the averaged accuracy of the algorithm based on $50$ random graph samples for the network size $n_1$ from $100$ to $500$ and $n_2=0.05n_1$.  
The link probabilities are set to be $\ell_{11} = 0.05$, $\ell_{12} = 0.1$, and $\ell_{22} = 0.5$. Additionally, we set the attention parameter $u = \bar{u}_1 + 0.01, \dots, \bar{u}_1 + 0.04$, where $\bar{u}_1$ is given in Theorem~\ref{thm_main}. Fig.~\ref{fig_eep_acc} shows that the averaged accuracy of Algorithm~\ref{alg_1} increases with the network size. Additionally, as $u$ grows, the performance becomes worse, indicating the impact of nonlinearity on the structure of the equilibria. To study the influence of the nonlinear interaction on detectability, we examine detection performance for the model~\eqref{eq_opinion_model_compactform} with the saturating function $S(x) = x/(1+|x|)$, $x/\sqrt{1+x^2}$, $\tanh (x)$, and $\tx{erf}(\sqrt{\pi}x/2)$, where $\tx{erf}$ is the Gauss error function.
Fig.~\ref{fig_eep_nonlinear} illustrates that the detection performance worsens as agents become more easily saturated.
Next, we compare Algorithm~\ref{alg_1} with classic methods which are applied directly to the adjacency matrices.
As shown in Fig.~\ref{fig_eep_comp}, the Girvan--Newman method (where the partition with two communities is selected) and the spectral clustering have high accuracy because the SBM has two blocks. 
For the Louvain method, we set its resolution parameter to produce a partition with two communities. It has difficulty identifying the smaller community, and hence does not perform well.
Utilizing only state observations instead of graphs, Algorithm~\ref{alg_1} can achieve high accuracy when $n$ is large enough while $u-\bar{u}_1$ is not too large.

For SSBM with positive weights, Algorithm~\ref{alg_1} cannot detect the communities w.h.p. This is shown in Fig.~\ref{fig_2ssbm_acc}, where the averaged accuracy of the algorithm is close to $0.5$, similar to random guess. This validates Theorem~\ref{thm_main}~(i.b). Here the attention parameter is chosen as above. The link probability within a community is $\ell_{\tx{s}} = 0.3$ and that between communities $\ell_{\tx{d}} = 0.05$.


\begin{figure*}
    \centering
    \subfigure[\label{fig_2ssbm_acc} Positive interpersonal influence. The accuracy is similar to random guess (close to $0.5$).]{\qquad\qquad \includegraphics[width=0.21\textwidth]{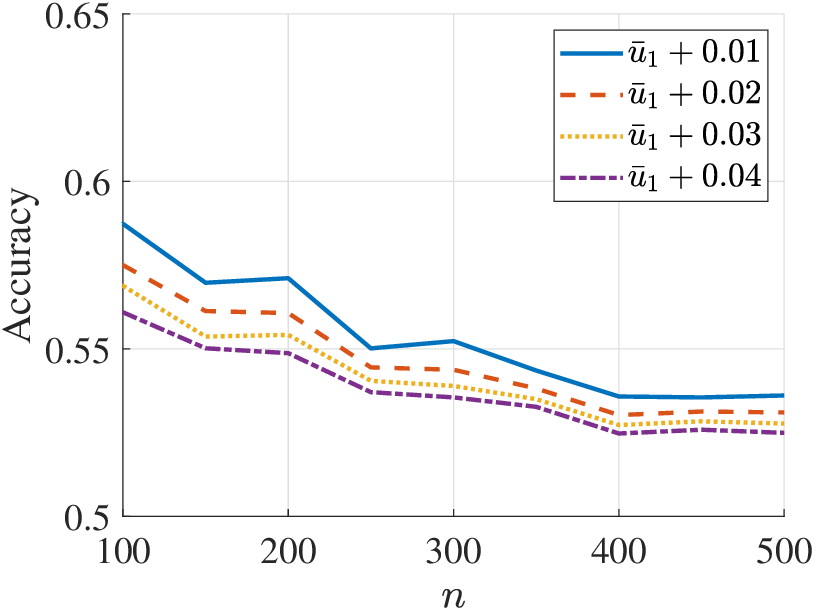}\qquad\qquad}\qquad
    \subfigure[\label{fig_neg_acc} The dissassortative SSBM with negative interpersonal influence. The accuracy increases with $n$.]{\qquad\qquad\includegraphics[width=0.21\textwidth]{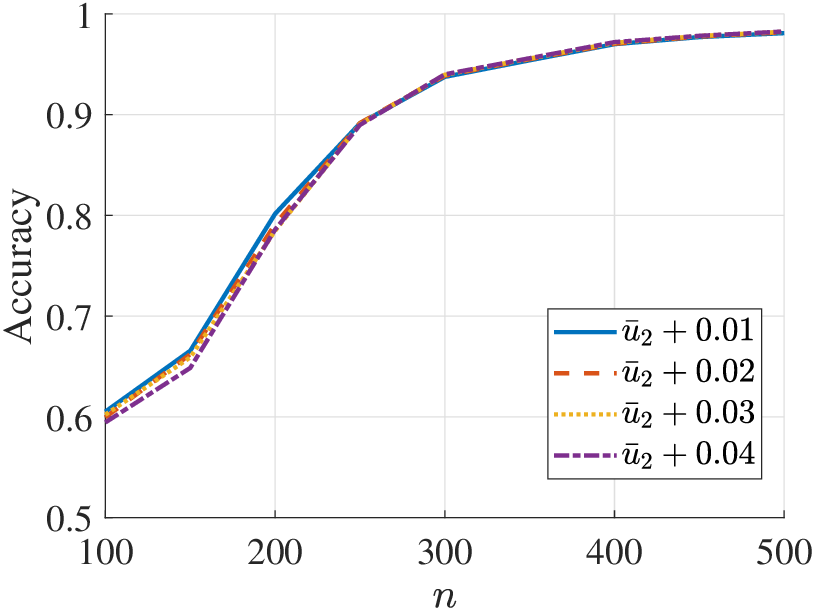} \qquad\qquad}
    \caption{Detection accuracy for the SSBM.}
\end{figure*}

\begin{figure*}[tbp]
    \centering
    \subfigure[\label{fig_multi_acc} The accuracy of Algorithm~\ref{alg_2}  increases with the sample size $m$ and the network size $n$. ]{\qquad\qquad~~ \includegraphics[width=0.21\textwidth]{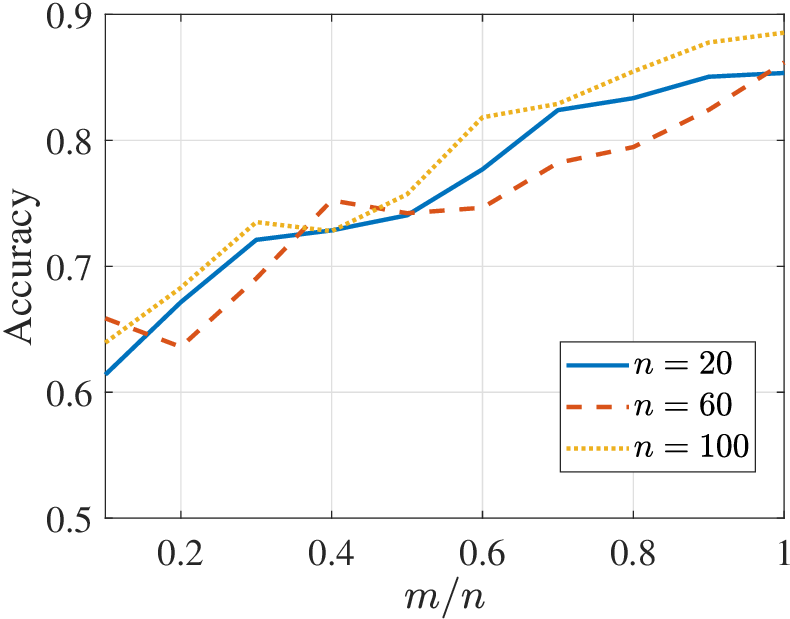} \qquad\qquad~~}\qquad
    \subfigure[\label{fig_multi_comp}The accuracy of Algorithm~\ref{alg_2} with $n=20$, the Louvain, the Girvan--Newman, and spectral clustering methods.]{\qquad\qquad ~~ \includegraphics[width=0.21\textwidth]{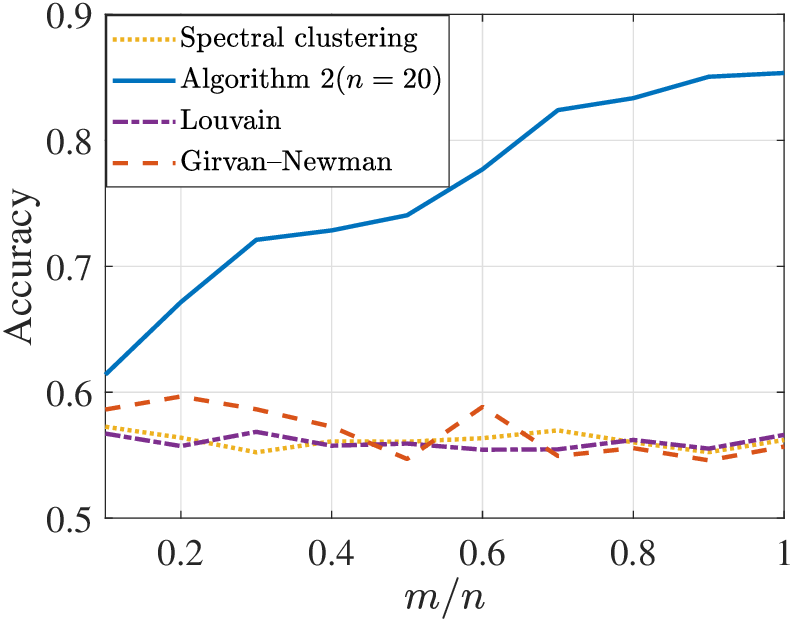} \qquad\qquad~~}
    \caption{\label{fig_multi}Detection accuracy from multiple input-equilibrium pairs.}
\end{figure*}

We then consider the disassortative SSBM with negative influence weights, to validate Theorem~\ref{thm_main}~(ii). In the experiment, we set $u = \bar{u}_2+0.01$, $\dots$, $\bar{u}_2 + 0.04$, and link probabilities to be $\ell_{\tx{s}} = 0.005$ and $\ell_{\tx{d}} = 0.03$. The link probabilities here are much smaller than previous ones, but the algorithm still has high detection accuracy and the attention parameter has less impact (see Fig.~\ref{fig_neg_acc}). 

Finally, we study the performance of Algorithm~\ref{alg_2} based on multiple input-equilibrium pairs $\{[\bfb^{(1)}~\bfx^{(1)}],$ $ \dots, [\bfb^{(m)}~\bfx^{(m)}]\}$. We consider positive influence weights and the same SSBM as the second experiment. The network size $n$ is set to be  $20,60$, and $100$, and $u$ to be $\bar{u}_1 + 0.01$. For each network, we generate $10$ graph samples, and for each graph sample we collect $10$ sets of input-equilibrium pairs where $\bfb^{(k)}$ are independently generated from the standard Gaussian distribution. To test the effect of sample size, we set $m = 0.1n$, $0.2n$, $\dots$, $n$. Fig.~\ref{fig_multi_acc} shows that the performance of the algorithm increases with both sample and network sizes. We then compare the proposed algorithm with classic methods for community detection from dynamical observations: the spectral clustering applied to the sample covariance matrix $\frac1m \sum_{k=1}^m (\bfx^{(k)} - \bar{\bfx}) (\bfx^{(k)} - \bar{\bfx})^\tp$ (e.g.,~\cite{wai2019blind,schaub2020blind}), the Girvan-Newman and the Louvain algorithms applied to correlation matrices~\cite{fenn2012dynamical}. The accuracy of these methods is below $0.6$, as shown in Fig.~\ref{fig_multi_comp}.
Algorithm~\ref{alg_2} performs better because it exploits the nonlinear structure and utilizes model parameter information.

\section{Conclusion}\label{sec_conclusion}
We studied community detection for a nonlinear opinion dynamics model over a stochastic block model. Two algorithms based on a single or multiple equilibria were investigated. Future work includes studying multi-dimensional generalizations~\cite{bizyaeva2023multi}, labeled SBMs~\cite{yun2014community}, and joint learning of communities and model parameters.

\appendix
\subsection{Proof of Theorem~\ref{thm_main}}\label{app_pf_thm}
\subsubsection{Proof of (i.a)}\label{app_pf_thm_ia} The proof is divided into three steps. First, we use the Lyapunov-Schmidt reduction~\cite{golubitsky1985singularities} to derive an explicit expression of the equilibrium $\bfx^*$. Then we apply concentration inequalities to approximate $\bfx^*$ by using the eigenvector corresponding to the largest eigenvalue of the expected adjacency matrix $\EE\{A\}$. Finally we obtain the conclusion by examining the properties of the eigenvector.

\textbf{Step 1.} This step is devoted to derive the following result.
\begin{lem}\label{lem_LS_reduction}
    Suppose that the graph $\mtcg = (\mtcv,\mtce,A)$ is connected, $
    \bfb = \bfo$, and $\bfx^*$ is an equilibrium emerging at the bifurcation from the origin as stated in Proposition~\ref{prop_bifurcation}. If $\gamma>0$, and $u-u_1>0$ is small enough, then $\bfx^* = c(u) \bfw^{\max}$, where $\bfw^{\max}$ is the unit eigenvector of $A$ corresponding to $\lambda_{\max}(A)$, and $c(u)$ depends on $u$ and $c(u) \to 0$ as $u\to u_1$.  
\end{lem}
\begin{proof}
    For an equilibrium $\bfx$ of the model~\eqref{eq_opinion_model_compactform} with $\bfb = \bfo$, it is a solution of the equation
    \begin{align}\label{eq_Phi}
        \bfPhi(\bfx,u) \Let - d \bfx + u \bfS(\alpha\bfx+\gamma A \bfx) = \bfo.
    \end{align}
    The Jacobian of $\bfPhi$ with respect to $\bfx$ at the origin is 
    \begin{align*}
        D_{\bfx} \bfPhi (\bfo, u) = (u\alpha - d) I + u\gamma A,
    \end{align*}
    and $J \Let D_{\bfx}\bfPhi(\bfo,u_1)$ has a single zero eigenvalue, so $\rank(D_{\bfx}\bfPhi(\bfo,u_1)) = n-1$. The idea of the Lyapunov-Schmidt reduction is to separately solve~\eqref{eq_Phi} near $(\bfo,u_1)$ for the corresponding nondegenerate $n-1$ variables of $\bfx$, and turn~\eqref{eq_Phi} into a reduced equation for the remaining unknown. 
    For symmetric $A$, there is orthogonal $[\tilde{W}~\bfw^{\max}]$ such that
    \begin{align*}
        \begin{bmatrix}
            \tilde{W}^\tp\\
            (\bfw^{\max})^\tp
        \end{bmatrix}
        A [\tilde{W}~\bfw^{\max}] &= 
        \begin{bmatrix}
            \tilde{W}^\tp A \tilde{W} & \bfo \\
            \bfo & \lambda_{\max}(A)
        \end{bmatrix}
    \end{align*}
    Let $E \Let I - \bfw^{\max} (\bfw^{\max})^\tp$ be the projection of $\RR^n$ onto $\range J$, and $\ker E = \operatorname{span} \bfw^{\max}$. Let $I - E$ be the complementary projection. Then the aforementioned decomposition can be written as follows
    \begin{align}\label{eq_Phi_a}
        E \bfPhi(\bfx,u) = \bfo, \\\label{eq_Phi_b}
        (I-E) \bfPhi(\bfx,u) = \bfo.
    \end{align}
    For $\bfx \in \RR^n$, it can be written as $\bfx = \bfw^{\max} x_m + \tilde{W} \tilde{\bfx}$, where $\tilde{\bfx} \in \RR^{n-1}$. Further, since $[\tilde{W} ~\bfw^{\max} ]^\tp E = [\tilde{W} ~\bfo ]^\tp$,~\eqref{eq_Phi_a} is equivalent to
    \begin{align}\label{eq_reduced_Phi_a}
        \bfo = \tilde{W}^\tp \bfPhi(\bfx,u) \teL \bfF(x_m,\tilde{\bfx},u).
    \end{align}
    Applying the implicit function theorem near $(0,\bfo,u_1)$ to~\eqref{eq_reduced_Phi_a} yields the dependence $\tilde{\bfx} = \tilde{\bfx}(x_m,u_1)$. Specifically, since
    \begin{align*}
        D_{x_m} \bfF(0,\bfo,u) &= \tilde{W}^\tp [-d I + u(\alpha I + \gamma A)] \bfw^{\max} = 0,\\
        D_{\tilde{x}} \bfF(0,\bfo,u) &= (u\alpha - d) I_{n-1} + u \gamma \tilde{W}^\tp A \tilde{W} \teL \tilde{J},\\
        D_{u} \bfF(0,\bfo,u) &= \tilde{W}^\tp \bfS(\bfo) = 0,
    \end{align*}
    and $\bfF$ is odd in $\bfx$, $\tilde{\bfx} = \tilde{J}^{-1} O((u-u_1)^3)$. Note that $\|\tilde{J}^{-1}\| = (\gamma \Theta(\lambda_{\max}(A) - \lambda_{n-1}(A)))^{-1}$ with $\lambda_{n-1}(A)$ the second largest eigenvalue, and is of order $\Theta(1)$ under the assumptions of the theorem. So $\tilde{\bfx} =  O((u-u_1)^3)$. As a result,~\eqref{eq_Phi_b} can be reduced to 
    \begin{align*}
        0 &= \bfw^{\max} (\bfw^{\max})^\tp [-d \bfx + u \bfS(\alpha \bfx + \gamma A \bfx)] \\
        &= \bfw^{\max} [- d x_m + u (\bfw^{\max})^\tp \bfS\big((\alpha I + \gamma A) \bfw^{\max} x_m\\
        &\quad~  + O((u-u_1)^3) \big) + O((u-u_1)^3)].
    \end{align*}
    From the Taylor expansion, $x_m = O((u-u_1)^2)$, which yields the conclusion by noticing $\bfx = \bfw^{\max} x_m + \tilde{W} \tilde{\bfx}$.
\end{proof}

\textbf{Step 2.} 
Now we use the eigenvector of $\EE\{A\}$ to approximate $\bfw^{\max}$ of $A$ generated from an SBM. The following lemma quantifies the deviation of $A$ generated from a random graph model and its eigenvalues from the expectation (Theorem~1 of~\cite{chung2011spectra}).
\begin{lem}
    Let $\mtcg = (\mtcv,\mtce,A)$ be a graph generated from a random graph model with independent edges and $|\mtcv| = n$, and $\Delta$ be the maximum expected degree of $\mtcg$. If $\Delta = \omega(\log n)$, then it holds w.h.p. for all $1\le i\le n$ that
    \begin{align*}
        |\lambda_i(A) - \lambda_i(\bar{A})| \le \|A-\bar{A}\| \le c \sqrt{\Delta \log n},
    \end{align*}
    where $c > 0$ is a constant and $\bar{A} \Let \EE\{A\}$. \hfill\QED
\end{lem}
Furthermore, the difference of the eigenspaces of $A$ from their expected counterparts can be bounded by using the Davis-Kahan theorem.
Applying~(4.19) of~\cite{vershynin2018high} to $A$ and $\bar{A}$, we know that, if $\min_{j<n} |\lambda_j(\bar{A}) - \lambda_{\max}(\bar{A})| \teL \delta > 0$, then 
\begin{align}
    \| \bfw^{\max} - \theta \bar{\bfw}^{\max} \| \le \frac{2^{3/2} \| A - \bar{A}\| }{\delta},
\end{align}
for some $\theta \in \{1,-1\}$, where $\bar{\bfw}^{\max}$ is the unit eigenvector of $\bar{A}$ corresponding to its eigenvalue $\lambda_{\max}(\bar{A})$.

If the random graph is $\tsbm(\bfn,\bm{\ell})$, it can be shown that $\EE\{A\} + \diag(\ell_{11}I_{n_1},\ell_{22}I_{n_2})$ has two nonzero eigenvalues:
\begin{align*}
    \bar{\lambda}_{\max} &\Let \frac12 \Big[(\ell_{11}n_1 + \ell_{22}n_2) \\
    &\quad~ + \sqrt{(\ell_{11}n_1 - \ell_{22}n_2)^2 + 4n_1n_2 \ell_{12}^2} \Big],\\
    \bar{\lambda}_{-} &\Let \frac12 \Big[(\ell_{11}n_1 + \ell_{22}n_2) \\
    &\quad~ - \sqrt{(\ell_{11}n_1 - \ell_{22}n_2)^2 + 4n_1n_2 \ell_{12}^2} \Big].
\end{align*}
When the SBM is symmetric, $\bar{\lambda}_{\max} = (\ell_{\tx{s}} + \ell_{\tx{d}})n/2$ and $\bar{\lambda}_- = (\ell_{\tx{s}} - \ell_{\tx{d}})n/2$. 
Note that $\max\{\ell_{11},\ell_{22}\} = o(\min\{\bar{\lambda}_{\max},\bar{\lambda}_-)\})$, so the eigenvalues of $\EE\{A\}$ is close to those of $\EE\{A\} + \diag(\ell_{11}I_{n_1},\ell_{22}I_{n_2})$ when $n$ is large enough. We summarize the findings in the following lemma.
\begin{lem}\label{lem_eigenvect}
    Suppose that Assumption~\ref{asmp_connect} holds. Then 
    \begin{align*}
        \| \bfw^{\max} - \theta \bar{\bfw}^{\max} \| = \Theta\Big(\frac{\sqrt{(\ell_{\tx{s}}+\ell_{\tx{d}})n\log n}}{(\ell_{\tx{d}} + \min\{\ell_{\tx{s}},\ell_{\tx{d}}\})n} \Big),
    \end{align*}
    $\exists \theta\in\{1,-1\}$, holds for $\tssbm(n,\bm{\ell})$. For $\tsbm(\bfn,\bm{\ell})$, the bound is $\Theta(\sqrt{\Delta \log n}/(\bar{\lambda}_{\max} - \max\{\bar{\lambda}_-,0\}))$. \hfill\QED
\end{lem}

\textbf{Step 3.} Under Assumption~\ref{asmp_connect}, $\mtcg$ is connected w.h.p., and $u_1 \sim d/(\alpha + \lambda_{\max}(\bar{A}) / \Delta) = \bar{u}_1$.  Under the assumption of~(i.a), it can be shown that $\bar{\lambda}_{\max} = \Theta(\Delta) = \omega(\log n)$. So the condition of Lemma~\ref{lem_eigenvect} holds. It suffices to study the structure of $\bar{\bfw}^{\max}$. Note that $\bar{\bfw}^{\max}$ has the structure $[w_1 \bfl_{n_1}^\tp~w_2 \bfl_{n_2}^\tp]^\tp$. From $\bar{A} \bar{\bfw}^{\max} = \bar{\lambda}_{\max} \bar{\bfw}^{\max}$, we have that
\begin{align*}
    \ell_{11}n_1 w_1 + \ell_{12}n_2 w_2 &= \bar{\lambda}_{\max} w_1,\\
    \ell_{21}n_1 w_1 + \ell_{22}n_2 w_2 &= \bar{\lambda}_{\max} w_2.
\end{align*}
Dividing the first equation by the second, it follows that 
\begin{align*}
    w_1 \sqrt{\frac{\ell_{22}}{\ell_{11}}} = \Theta(w_2).
\end{align*}
From $\|\bar{\bfw}^{\max}\| = 1$, $w_1 = \Theta(1/\sqrt{n_1}) = \Theta(1/\sqrt{n})$ and $w_2 = \Theta(1/\sqrt{n_2}) = \omega(1/\sqrt{n})$. Note that the entry-wise concentration error is of order $O(1/\sqrt{n\Delta^{c_0}})$ with $c_0 \in (0,1)$ from Lemma~\ref{lem_eigenvect}. So $\bfx^*$ consists of two clusters w.h.p. and Algorithm~\ref{alg_1} can achieve almost exact recovery.

\subsubsection{Proof of (i.b)} 
For $\tssbm(n,\bm{\ell})$, the expected adjacency matrix has a block structure
\begin{align*}
    \EE\{A\} \sim \begin{bmatrix*}
        \ell_{\tx{s}}\bfl_{n/2,n/2} & \ell_{\tx{d}}\bfl_{n/2,n/2} \\ 
        \ell_{\tx{d}}\bfl_{n/2,n/2} & \ell_{\tx{s}}\bfl_{n/2,n/2}
    \end{bmatrix*},
\end{align*}
so $\bar{\bfw}^{\max} \sim \bfl_n /\sqrt{n}$. Combining Lemmas~\ref{lem_LS_reduction} and~\ref{lem_eigenvect} yields that 
\begin{align*}
    \bfx^* &= c(u) \bfw^{\max}  \\
    &= c(u) \theta \bar{\bfw}^{\max} + c(u) (\bfw^{\max} - \theta \bar{\bfw}^{\max}) \\
    &= c(u) \theta \bar{\bfw}^{\max} + o(c(u)),
\end{align*}
which proves the first statement of (i.b). 

Now suppose that Algorithm~\ref{alg_1} can achieve almost exact recovery. Then there must exists $\chi^*_1$ and $\chi^*_2$ such that $\chi_1^*,\chi_2^* = \Theta(1/\sqrt{n})$, $|\bfx^*_i - \chi^*_1| < \eps$ for $i\in \mtcv_1$, and $|\bfx^*_j - \chi^*_2| < \eps$ for $j\in \mtcv_2$ for some $\eps>0$ with $3\eps < |\chi^*_1-\chi^*_2|$, except for $o(n)$ agents. Assuming without loss of generality that $\chi^*_1 >\chi^*_2>0$ and letting $\theta=1$, we have that
\begin{align*}
    \bar{\bfw}^{\max} = \frac{\bfx^*}{c(u)} + o(1).
\end{align*}
Let $\tilde{w} \Let  \bar{A} \bar{\bfw}^{\max}$, and we have
\begin{align*}
    \tilde{w}_i \le \frac{n}{2} [\ell_{\tx{s}} \chi_1^* + \ell_{\tx{d}} \chi_2^* + \eps (\ell_{\tx{s}} + \ell_{\tx{d}}) + o(\chi_k^*)(\ell_{\tx{s}} + \ell_{\tx{d}})], i\in \mtcv_1,\\
    \tilde{w}_j \ge \frac{n}{2} [\ell_{\tx{d}} \chi_1^* + \ell_{\tx{s}} \chi_2^* - \eps (\ell_{\tx{s}} + \ell_{\tx{d}}) + o(\chi_k^*)(\ell_{\tx{s}} + \ell_{\tx{d}})], j\in \mtcv_2,
\end{align*}
except for $o(n)$ agents. Then for $i\in \mtcv_1$ and $j\in\mtcv_2$
\begin{align*}
    &\tilde{w}_j - \tilde{w}_i \\
    &\ge \frac{n}{2} [(\ell_{\tx{d}} - \ell_{\tx{s}}) (\chi_1^* - \chi_2^*) - 2\eps(\ell_{\tx{s}} + \ell_{\tx{d}}) + o(\chi_k^*)(\ell_{\tx{s}} + \ell_{\tx{d}})] \\
    &\ge \frac{n}{2} \Big[\frac13 (\chi_1^*-\chi_2^* )(\ell_{\tx{d}} - 3\ell_{\tx{s}}) +o(\chi_k^*)(\ell_{\tx{s}} + \ell_{\tx{d}})  \Big] > 0,
\end{align*}
which contradicts that $\tilde{w}_j < \tilde{w}_i$ from the definition of eigenvectors. Hence the assertion is false and Algorithm~\ref{alg_1} cannot achieve almost exact recovery.

\subsubsection{Proof of (ii)}
It suffices to show that $\bfx^*$ form two clusters similar to (i.a). When $\gamma < 0$, $u_2 \sim d/(\alpha + \gamma \lambda_{\min}(\EE\{A\})) = \bar{u}_2$. For $\tssbm(n,\bm{\ell})$ with $\ell_{\tx{s}} < \ell_{\tx{d}}$, the smallest eigenvalue of $\bar{A}$ is $\bar{\lambda}_- = (\ell_{\tx{s}} - \ell_{\tx{d}})n/2$ and the corresponding eigenvector is $\bfw^{\min} = [\bfl_{n/2}^\tp/\sqrt{n},~-\bfl_{n/2}^\tp/\sqrt{n}]^\tp$. Again note that the concentration error is of order $O(1/\sqrt{n\Delta^{c_0}})$, so Algorithm~\ref{alg_1} can achieve almost exact recovery in this case.

\addtolength{\textheight}{-12cm}   


\bibliographystyle{ieeetr}
\bibliography{bibliography}

\begin{thebibliography}{10}

\bibitem{ramakrishna2020user}
R.~Ramakrishna, H.-T. Wai, and A.~Scaglione, ``A user guide to low-pass graph signal processing and its applications: {T}ools and applications,'' {\em IEEE Signal Processing Magazine}, vol.~37, no.~6, pp.~74--85, 2020.

\bibitem{ravazzi2021learning}
C.~Ravazzi, F.~Dabbene, C.~Lagoa, and A.~V. Proskurnikov, ``Learning hidden influences in large-scale dynamical social networks: A data-driven sparsity-based approach, in memory of {R}oberto {T}empo,'' {\em IEEE Control Systems Magazine}, vol.~41, no.~5, pp.~61--103, 2021.

\bibitem{kempe2003maximizing}
D.~Kempe, J.~Kleinberg, and {\'E}.~Tardos, ``Maximizing the spread of influence through a social network,'' in {\em Proceedings of the Ninth ACM SIGKDD International Conference on Knowledge Discovery and Data Mining}, pp.~137--146, 2003.

\bibitem{burke2002hybrid}
R.~Burke, ``Hybrid recommender systems: {S}urvey and experiments,'' {\em User Modeling and User-Adapted Interaction}, vol.~12, pp.~331--370, 2002.

\bibitem{fortunato2016community}
S.~Fortunato and D.~Hric, ``Community detection in networks: A user guide,'' {\em Physics Reports}, vol.~659, pp.~1--44, 2016.

\bibitem{wai2019blind}
H.-T. Wai, S.~Segarra, A.~E. Ozdaglar, A.~Scaglione, and A.~Jadbabaie, ``Blind community detection from low-rank excitations of a graph filter,'' {\em IEEE Transactions on Signal Processing}, vol.~68, pp.~436--451, 2019.

\bibitem{hoffmann2020community}
T.~Hoffmann, L.~Peel, R.~Lambiotte, and N.~S. Jones, ``Community detection in networks without observing edges,'' {\em Science Advances}, vol.~6, no.~4, p.~eaav1478, 2020.

\bibitem{xing2023community}
Y.~Xing, X.~He, H.~Fang, and K.~H. Johansson, ``Community structure recovery and interaction probability estimation for gossip opinion dynamics,'' {\em Automatica}, vol.~154, p.~111105, 2023.

\bibitem{baumann2020modeling}
F.~Baumann, P.~Lorenz-Spreen, I.~M. Sokolov, and M.~Starnini, ``Modeling echo chambers and polarization dynamics in social networks,'' {\em Physical Review Letters}, vol.~124, no.~4, p.~048301, 2020.

\bibitem{bizyaeva2022nonlinear}
A.~Bizyaeva, A.~Franci, and N.~E. Leonard, ``Nonlinear opinion dynamics with tunable sensitivity,'' {\em IEEE Transactions on Automatic Control}, vol.~68, no.~3, pp.~1415--1430, 2022.

\bibitem{blondel2008fast}
V.~D. Blondel, J.-L. Guillaume, R.~Lambiotte, and E.~Lefebvre, ``Fast unfolding of communities in large networks,'' {\em Journal of Statistical Mechanics: Theory and Experiment}, vol.~2008, no.~10, p.~P10008, 2008.

\bibitem{rosvall2008maps}
M.~Rosvall and C.~T. Bergstrom, ``Maps of random walks on complex networks reveal community structure,'' {\em Proceedings of the National Academy of Sciences}, vol.~105, no.~4, pp.~1118--1123, 2008.

\bibitem{yun2014community}
S.-Y. Yun and A.~Proutiere, ``Community detection via random and adaptive sampling,'' in {\em Conference on Learning Theory}, pp.~138--175, 2014.

\bibitem{massoulie2014community}
L.~Massouli{\'e}, ``Community detection thresholds and the weak ramanujan property,'' in {\em Proceedings of the Forty-Sixth Annual ACM Symposium on Theory of Computing}, pp.~694--703, 2014.

\bibitem{abbe2017community}
E.~Abbe, ``Community detection and stochastic block models: {R}ecent developments,'' {\em The Journal of Machine Learning Research}, vol.~18, no.~1, pp.~6446--6531, 2017.

\bibitem{prokhorenkova2022less}
L.~Prokhorenkova, A.~Tikhonov, and N.~Litvak, ``When less is more: Systematic analysis of cascade-based community detection,'' {\em ACM Transactions on Knowledge Discovery from Data}, vol.~16, no.~4, pp.~1--22, 2022.

\bibitem{peixoto2019network}
T.~P. Peixoto, ``Network reconstruction and community detection from dynamics,'' {\em Physical Review Letters}, vol.~123, no.~12, p.~128301, 2019.

\bibitem{schaub2020blind}
M.~T. Schaub, S.~Segarra, and J.~N. Tsitsiklis, ``Blind identification of stochastic block models from dynamical observations,'' {\em SIAM Journal on Mathematics of Data Science}, vol.~2, no.~2, pp.~335--367, 2020.

\bibitem{wai2022community}
H.-T. Wai, Y.~C. Eldar, A.~E. Ozdaglar, and A.~Scaglione, ``Community inference from partially observed graph signals: Algorithms and analysis,'' {\em IEEE Transactions on Signal Processing}, vol.~70, pp.~2136--2151, 2022.

\bibitem{xing2023almost}
Y.~Xing and K.~H. Johansson, ``Almost exact recovery in gossip opinion dynamics over stochastic block models,'' in {\em IEEE Conference on Decision and Control}, pp.~2421--2426, 2023.

\bibitem{eldar2013effects}
E.~Eldar, J.~D. Cohen, and Y.~Niv, ``The effects of neural gain on attention and learning,'' {\em Nature Neuroscience}, vol.~16, no.~8, pp.~1146--1153, 2013.

\bibitem{bogacz2007extending}
R.~Bogacz, M.~Usher, J.~Zhang, and J.~L. McClelland, ``Extending a biologically inspired model of choice: Multi-alternatives, nonlinearity and value-based multidimensional choice,'' {\em Philosophical Transactions of the Royal Society B: Biological Sciences}, vol.~362, no.~1485, pp.~1655--1670, 2007.

\bibitem{proskurnikov2017tutorial}
A.~V. Proskurnikov and R.~Tempo, ``A tutorial on modeling and analysis of dynamic social networks. {P}art {I},'' {\em Annual Reviews in Control}, vol.~43, pp.~65--79, 2017.

\bibitem{leonard2021nonlinear}
N.~E. Leonard, K.~Lipsitz, A.~Bizyaeva, A.~Franci, and Y.~Lelkes, ``The nonlinear feedback dynamics of asymmetric political polarization,'' {\em Proceedings of the National Academy of Sciences}, vol.~118, no.~50, p.~e2102149118, 2021.

\bibitem{bizyaeva2021control}
A.~Bizyaeva, T.~Sorochkin, A.~Franci, and N.~E. Leonard, ``Control of agreement and disagreement cascades with distributed inputs,'' in {\em IEEE Conference on Decision and Control}, pp.~4994--4999, 2021.

\bibitem{bizyaeva2021patterns}
A.~Bizyaeva, A.~Matthews, A.~Franci, and N.~E. Leonard, ``Patterns of nonlinear opinion formation on networks,'' in {\em American Control Conference}, pp.~2739--2744, 2021.

\bibitem{bizyaeva2022switching}
A.~Bizyaeva, G.~Amorim, M.~Santos, A.~Franci, and N.~E. Leonard, ``Switching transformations for decentralized control of opinion patterns in signed networks: Application to dynamic task allocation,'' {\em IEEE Control Systems Letters}, vol.~6, pp.~3463--3468, 2022.

\bibitem{fenn2012dynamical}
D.~J. Fenn, M.~A. Porter, P.~J. Mucha, M.~McDonald, S.~Williams, N.~F. Johnson, and N.~S. Jones, ``Dynamical clustering of exchange rates,'' {\em Quantitative Finance}, vol.~12, no.~10, pp.~1493--1520, 2012.

\bibitem{bizyaeva2023multi}
A.~Bizyaeva, A.~Franci, and N.~E. Leonard, ``Multi-topic belief formation through bifurcations over signed social networks,'' {\em arXiv preprint arXiv:2308.02755}, 2023.

\bibitem{golubitsky1985singularities}
M.~Golubitsky, I.~Stewart, and D.~G. Schaeffer, {\em Singularities and Groups in Bifurcation Theory: Volume I}.
\newblock Springer Science \& Business Media, 1985.

\bibitem{chung2011spectra}
F.~Chung and M.~Radcliffe, ``On the spectra of general random graphs,'' {\em The Electronic Journal of Combinatorics}, p.~P215, 2011.

\bibitem{vershynin2018high}
R.~Vershynin, {\em High-Dimensional Probability: An Introduction with Applications in Data Science}.
\newblock Cambridge University Press, 2018.

\end{thebibliography}

\end{document}